\documentclass{llncs}

\usepackage{amssymb}
\usepackage{amsmath}
\usepackage{graphicx}

\title{A generalization of Hopcroft-Karp algorithm for semi-matchings and covers in bipartite graphs}
\titlerunning{A generalization of Hopcroft-Karp algorithm for semi-matchings and covers in bipartite graphs}  

\author{J\'an Katreni{\v c} \and Gabriel Semani{\v s}in}

\authorrunning{J\'an Katreni{\v c} and Gabriel Semani{\v s}in}   
\institute{
Institute of Computer Science,\\
P.J.~{\v S}af{\'a}rik University, Faculty of Science,\\
Jesenn\'a~5, 041 54 Ko{\v s}ice, Slovak Republic,\\
\email{jan.katrenic@upjs.sk, gabriel.semanisin@upjs.sk}}

\begin{document}

\maketitle 

\begin{abstract}

An $(f,g)$-semi-matching in a bipartite graph $G=(U \cup V,E)$ is a set of edges $M \subseteq E$ such that each vertex $u\in U$ is incident with at most $f(u)$ edges of $M$, and each vertex $v\in V$ is incident with at most $g(v)$ edges of $M$. 
In this paper we give an algorithm that for a graph with $n$ vertices and $m$ edges, $n\le m$, constructs a maximum $(f,g)$-semi-matching in running time $O(m\cdot \min (\sqrt{\sum_{u\in U}f(u)}, \sqrt{\sum_{v\in V}g(v)}) )$.
Using the reduction of \cite{WG2011}
our result on maximum $(f,g)$-semi-matching problem directly implies an algorithm for the optimal semi-matching problem with running time $O(\sqrt{n}m \log n)$.
\end{abstract}

\section{Introduction}

We consider finite non-oriented graphs without loops and multiple edges. In general we use standard  concepts and notation  of  graph theory. In particular,  $deg(u)$ denotes the degree of a vertex $u$ in $G=(V,E)$.
If $M\subseteq E$ then  $deg_M(u)$ denotes the number of edges of $M$ incident with $u$.
If $f$ is and integer valued function defined for all vertices of $G$ and $X\subseteq V$ then $f(X)$ stands for the sum $\sum_{v\in X} f(v)$. 

Let $G = (U \cup V, E)$ be a bipartite graph with $n = |U|+|V|$ vertices and $m = |E|$ edges (throughout the paper we consider only non-trivial case  with no isolated vertices, i.e. $n-1 \le m$). A {\em semi-matching} $M$ of $G$ is a set of edges $M\subseteq E(G)$, such that each vertex of $U$ is incident with exactly one edge of $M$.

Semi-matching is a natural generalization of the classical matching in bipartite graphs. Although the name of semi-matching was introduced recently in \cite{HarveyLLT06}, semi-matchings appear in many problems and were studied as  early as 1970s \cite{Horn} with applications in wireless sensor networks \cite{A3,A10,A6,A2,A1} and a wide area of scheduling problems \cite{Bruno,A5,A9,A7,A8}. For a weighted case of the problem we refer to \cite{ICALP10,A5,LeeLP09,HaradaYuta2007}.

The problem of finding an {\em optimal semi-matching} (see \cite{HarveyLLT06}) is motivated by the following off-line load balancing scenario: 
Given a set of tasks and a set of machines, each of which can process a subset of tasks. 
Each task requires one unit of processing time and must be assigned to some machine that can process it. The tasks have to be assigned in a manner that minimizes given optimization objective. One natural goal is to process all tasks with the minimum total completion time. Another goal is to minimize the average completion time, or total flow time, which is the sum of time units necessary for completion of all jobs (including the units while a job is waiting in the queue).

Let $M$ be a semi-matching. 
The {\em cost} of $M$, denoted by $cost(M)$, is defined as follows:
 $$cost(M) = \sum_{v \in V}  \frac{deg_M(v)\cdot(deg_M(v)+1)}2.$$

A semi-matching is {\em optimal}, if its  $cost$ is the smallest one among the costs of all admissible semi-matchings.
The problem of computing an optimal semi-matching was firstly studied by Horn \cite{Horn} and Bruno et al.\ \cite{Bruno} where an $O(n^3)$ algorithm was presented. The problem received considerable attention in the past few years. Harvey et al. \cite{HarveyLLT06} showed that by minimizing $cost$ of a semi-matching one minimizes simultaneously the maximum number of tasks assigned to a machine, the flow time and the variance of loads.
The same authors provided also a characterization of an optimal assignment based on cost-reducing paths and an algorithm for finding an optimal semi-matching in time $O(n\cdot m)$. It constructs an optimal semi-matching step by step starting with an empty semi-matching and in each iteration finds an augmenting path from a free $U$-vertex to a vertex in $V$ with the smallest possible degree.

The semi-matchings were generalized to the quasi-matchings by Bokal et al. \cite{Bresar09}. 
They consider an integer valued function $g$ defined on the vertex set and require that each vertex $v \in V$ is connected to at least $g(v)$ vertices of $U$.
 
An $(f,g)$-quasi-matching in a bipartite graph $G=(U \cup V,E)$ is a set of edges $M \subseteq E$ such that each vertex $u\in U$ is incident with at most $f(u)$ edges of $M$, and each vertex $v\in V$ is incident with at least $g(v)$ edges of $M$.
The authors provided a property of lexicographically minimum $g$-quasi-matching and showed that the lexicographically minimum $1$-quasi-matching equals to an optimal semi-matching. Moreover they also designed an algorithm to compute an optimal (lexicographically minimum) $g$-quasi-matching in running time $O(m \cdot g(V))$.

Similarly, in \cite{Bresar09} was defined an $(f,g)$-{\em semi-matching} of $G = (U \cup V, E)$, which 
is a set of edges $M\subseteq E$ such that every element $u$ of $U$ has at most $f(u)$ incident edges from $M$, and every element $v$ of $V$ has at most $g(v)$ incident edges from $M$. A {\em maximum $(f,g)$-semi-matching} is the one with as many edges as possible.

The complexity bound for computing an optimum semi-matching was further improved by Fakcharoenphol et al. \cite{ICALP10}, who presented $O(\sqrt{n} \cdot m \cdot \log n)$ algorithm for the optimal semi-matching problem. The algorithm uses a reduction to the min-cost flow problem and exploits the structure of the graphs and cost functions for an elimination of many negative cycles in a single iteration. 

Recently, in \cite{WG2011} it was presented a reduction from the optimum semi-matching problem to the maximum $(f,g)$-semi-matching, which shows that an optimal semi-matching of $G$
can be computed in time $O((n + m + T_{BDSM}(n, m)) \cdot \log{n})$ where $n = |U| + |V|$, $m = |E|$, and $T_{BDSM}(n, m)$ is the time complexity of an algorithm for computing a maximum $(f,1)$-semi-matching with $f(U) \le 2n$. 
By a result of \cite{MuchaS04}, the algorithm designed in \cite{WG2011} yields to a randomized algorithm for optimal semi-matching with a running time of $O(n^{\omega})$, where $\omega$ is the exponent of the best known matrix multiplication algorithm. Since $\omega \leq 2.38$, this algorithm broke through $O(n^{2.5})$ barrier for computing optimal semi-matching in dense graphs \cite{WG2011}.

In this paper we present an algorithm for finding a maximum $(f,g)$-semi-matching in running time 
$O(m\cdot \min\{\sqrt{ f(U)},  \sqrt{ g(V)} \})$. 
For the problem of computing an $(f,g)$-quasi-matching it gives an algorithm with running time $O(m \sqrt{g(V)})$.
For the  maximum $(f,1)$-semi-matching we get an complexity upper bound $O(\sqrt{n} \cdot m)$, which implies a bound $O(\sqrt{n} \cdot m \cdot \log n)$ for computing an optimal semi-matching of the algorithm presented in \cite{WG2011}.

\section{Augmenting paths and $(f,g)$-semi-matchings}

In this chapter we introduce concepts that will be used throughout the remaining part of the paper.

\begin{definition} 
Let $f: U \rightarrow \mathbb{N}$ and $g: V \rightarrow \mathbb{N}$ be mappings.
An $(f,g)$-semi-matching in a bipartite graph $G=(U \cup V,E)$ is a set of edges $M \subseteq E$ such that $deg_M(u)\le f(u)$ for each vertex $u\in U$, and $deg_M(v)\le g(v)$ for each vertex $v\in V$.
\end{definition}

\begin{definition}
An $(f,g)$-semi-matching $M$ of a graph $G=(U \cup V ,E)$ is called {\em maximum}, if for each $(f,g)$-semi-matching $M'$ of $G$ holds $|M|\ge|M'|$.
An $(f,g)$-semi-matching $M$ is called {\em perfect}, if $|M|=f(U)$.
\end{definition}
Note, that $(1,1)$-semi-matching is a matching in a bipartite graph.

\begin{definition}
Let $G=(U \cup V,E)$ be a bipartite graph and $H\subseteq E$. A path $P$ is called an $H$-{\em alternating path}, if each internal vertex of $P$ is incident with exactly one edge of $H \cap P$.
\end{definition}

\begin{definition}
Let $G=(U \cup V,E)$ be a bipartite graph and $H\subseteq E$. An $H$-{\em augmenting path} $P$ is an alternating path with the first and last vertex of $P$ not incident with an edge of $H \cap P$.
\end{definition}

\begin{definition}
Let $G=(U \cup V,E)$ be a bipartite graph, $H\subseteq E$, $P$ be an $H$-alternating path and $E(P)$ be the edge set of $P$.
We define an operator $\oplus$ as follows:
$$H \oplus P =  (H \cup E(P)) \setminus (E(P) \cap H).$$
\end{definition}
 
The next theorem provides a characterisation of maximum $(f,g)$-semi-matching. 
 
\begin{theorem} \label{lemma1}
 Let $M$ and $M'$ be an $(f,g)$-semi-matching of a graph $G$, $|M'|>|M|$. Then there exists an $M$-augmenting path $P$ with endvertices 
 $u\in U, v\in V$, $deg_M(u) < f(u)$ and $deg_M(v) < g(v)$ such that $E(P)\subseteq M\cup M'$.
\end{theorem}

\begin{proof}
We proceed by an induction on the size of $|M|$. Evidently, the assertion of the theorem is true for the smallest cases. Now, we may assume that $M \cap M' = \emptyset$, otherwise the assertion follows from the induction hypothesis. Let us put
$$A = \{ v\in V: deg_M(v) < deg_{M'}(v) \} .$$ 

Let $V_A$ be the set of vertices of $V$ for which there exists an $M$-alternating path starting in a vertex of $A$ with and edge of $M'$. Here a path of length $0$ is considered to be an $M$-alternating path, therefore $A\subseteq V_A$.

Let $U_A$ be the set of vertices of $U$ for which there exists an $M$-alternating path starting in a vertex of $A$ with an edge of $M'$. 

Let us put $V_B = V \setminus V_A$ and $U_B = U \setminus U_A$.
For sets  $X\subseteq U$ and $Y\subseteq V$ we introduce parameters $m(X,Y) = |E(G[X\cup Y]) \cap M|$ and
$m'(X,Y) = |E(G[X\cup Y]) \cap M'|$.

From the definition of $V_B$ we get $m(U_A,V_B) = 0$ and the definition of $U_A$ yields $m'(U_B,V_A) = 0$ (otherwise the existence of such an edge implies an existence of an $M$-alternating path starting at a vertex of $A$ by edge of $M'$). This is depicted on Figure~\ref{fig1}.

\begin{figure}[ht]
 \begin{center} 
 		\includegraphics[scale=0.8]{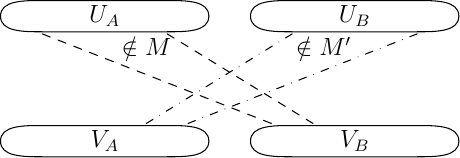} 
 \caption{The vertices of $G$ are divided into 4 parts. The edges between $U_B$ and $V_A$ cannot belong to $M'$, the edges between $U_A$ and $V_B$ cannot belong to $M$.}
\label{fig1}
\end{center}
\end{figure}

\noindent Since $|M| < |M'|$, we have $m(U,V) < m'(U,V)$. Moreover $m(U_A,V_B) = 0$ and $m'(U_B,V_A) = 0$ which gives
\begin{equation} \label{eq:1}
m(U_A,V_A) + m(U_B,V_A) + m(U_B,V_B) < m'(U_A,V_A) + m'(U_A,V_B) + m'(U_B,V_B)
\end{equation}

\noindent Since $A \cap V_B = \emptyset$ and $m(U_A,V_B) = 0$, we get the inequality
\begin{equation} \label{eq:2}
m(U_B,V_B) \ge m'(U_A,V_B) + m'(U_B,V_B).
\end{equation}

\noindent By \eqref{eq:1} and \eqref{eq:2} we get
\begin{equation}\label{eq:3}
m(U_A,V_A) + m(U_B,V_A) < m'(U_A,V_A). 
\end{equation}

\noindent Trivially, we have the following
\begin{equation}\label{eq:4}
 m(U_B,V_A) \ge - m'(U_A,V_B).
\end{equation}

\noindent Combining  \eqref{eq:3} and \eqref{eq:4} we obtain
\begin{equation} \label{eq:5}
m(U_A,V_A) < m'(U_A,V_A) + m'(U_A,V_B).
\end{equation}

From the inequality \eqref{eq:5} we can conclude that $U_A$ contains a vertex $u$ with $deg_M(u) < deg_{M'}(u)$. By the definition of $U_A$, it implies an existence of an $M$-augmenting path with endvertex $u$ and an endvertex from $A$.

\end{proof}

\begin{theorem} \label{thmmaximum}
A $(f,g)$-semi-matching $M$ of a graph $G = (U \cup V, E)$ is maximum if and only if there exists no $M$-augmenting path $P$ with endvertices $u\in U, v\in V$, $deg_M(u) < f(u)$ and $deg_M(v) < g(v)$.
\end{theorem}

\begin{proof} 
 Suppose to the contrary that there is a maximum $(f,g)$-semi-matching $M$ and  $M$-augmenting path $P$ with endvertices
 $u\in U, v\in V$ and  $deg_M(u) < f(u)$, $deg_M(v) < g(v)$. Then obviously
$M\oplus P$ is an $(f,g)$-semi-matching with $|M\oplus P| > |M|$.

The opposite direction comes from Theorem~\ref{lemma1}.

\end{proof}

The next theorem provides more information about the structure of augmenting paths.

\begin{theorem} \label{edgedis}
Let $M$ and $M'$ be $(f,g)$-semi-matchings of a bipartite graph $G$ such that $|M'|-|M|= k >0$. Then there exist $k$ edge-disjoint $M$-augmenting paths $P_1,P_2,\dots, P_k$ such that $M\oplus P_1 \oplus \dots \oplus P_k = M'.$
\end{theorem}

\begin{proof}
We prove the theorem by induction on the size of the graph $G$. The assertion obviously holds for the smallest possible cases.
If $M \cap M' \neq \emptyset$, then $G \setminus (M \cap M')$ and $M \setminus M'$, $M' \setminus M$ is an instance of theorem of smaller size and the claim follows from induction hypothesis.

Suppose now $M \cap M' = \emptyset$.
Using Theorem \ref{lemma1}, there exists an $M$-augmenting path $P$ such that its edges alternatively belongs to $M'$ and $M$. Therefore $|M' \setminus E(P)| - |M \setminus E(P) | = k - 1$ and $(M\oplus P) \cap E(P) = M' \cap E(P)$. 
Consider now the graph $G \setminus E(P)$ and edge sets $M\setminus E(P)$, $M' \setminus E(P)$. 
From the induction hypothesis there exist $k-1$ edge disjoint paths $P_1, \dots, P_{k-1}$ such that $(M \setminus E(P)) \oplus P_1 \oplus \dots P_{k-1} = (M' \setminus E(P))$.
Clearly, $P$ is edge disjoint with $P_1,\dots, P_{k-1}$ and

\noindent
\begin{eqnarray*}
M' & = & (M' \cap E(P)) \cup (M' \setminus E(P))\\
   & = & ((M\oplus P) \cap E(P) ) \cup ((M \setminus E(P)) \oplus P_1 \oplus \dots P_{k-1})\\
   & = & M \oplus P_1 \oplus \dots P_{k-1} \oplus P.
\end{eqnarray*}

\end{proof}

%

\begin{proof}

\begin{corollary} \label{semidis1}
 Let $M$ and $M'$ be an $(f,1)$-semi-matchings of a bipartite graph $G$ such that $|M'|-|M|=k>0$. Then there exist $k$ $M$-augmenting paths $P_1,P_2,\dots, P_k$ such that $M' = M \oplus P_1 \oplus \dots \oplus P_k$ and 
$E(P_i) \cap E(P_j) = \emptyset$, for each $i,j \in \{1,2,\dots, k\}, i \neq j$.
\end{corollary}
It  follows from Theorem \ref{edgedis} and the fact $deg_M(v) \le 1$, $v \in V$ that no two of those $M$-augmenting paths may overlap in a vertex $v\in V$. 
\end{proof}

Let $M$ be an $(f,g)$-semi-matching of a bipartite graph $G=(U \cup V,E)$. 
Denote by $V^g_M = \{ v\in V: deg_M(v) < g(v)\}$. We set $adist_M(x)$ to be the length of a shortest $M$-alternating path starting in any vertex of $V^g_M$ and ending in $x$. If no such $M$-alternating path exists, we put $adist_M(x)=+\infty$.

\begin{theorem}\label{nondecrease}
 Let $M$ be an $(f,g)$-semi-matching of a bipartite graph $G=(U \cup V,E)$ and $P$ be a shortest $M$-augmenting path.
 Then $adist_M(x) \le adist_{M\oplus P}(x)$ for each vertex $x\in {U\cup V}$.
\end{theorem}

\begin{proof}
Assume to the contrary that there exists at least one vertex $x$ such that $adist_M(x)>adist_{M\oplus P}(x)$. Let us choose such a vertex $x$ with the smallest possible value of $adist_M(x)$. It means that for each vertex $y$ with $adist_M(y)<adist_M(x)$ the inequality $adist_M(y) \le adist_{M\oplus P}(x)$ is valid.

Clearly $adist_{M\oplus P}(x)$ cannot be $0$, because in such a case $x$ is a vertex of $V$ for which $deg_{M\oplus P}(x)< g(x)$ and that is why $adist_M(x)$ must be zero as well.

Thus, $adist_{M\oplus P}(x)$ is at least $1$. Let $y$ be the predecessor of $x$ in a shortest $(M\oplus P)$-alternating path starting in a vertex of $V^g_{M\oplus P}$. Obviously $adist_{M\oplus P}(y)+1 = adist_{M\oplus P}(x)$. It also holds that $adist_M(y) \le adist_{M\oplus P}(y)$ (otherwise $x$ was not chosen correctly), what together with the previous equation gives $adist_M(y)<adist_{M\oplus P}(x)$. 
Together with the initial inequality for $y$  we obtain $adist_M(y)< adist_{M\oplus P}(x) < adist_{M}(x)$.
This implies that the edge $xy$ was changed, i.e. $xy \in P$ (otherwise the edge $xy$ could be used to violate the inequality $adist_M(v) > adist_{M\oplus P}(v)$). Let us distinguish  now two cases:

\smallskip
\noindent
Case1. $x\in U$ and $y\in V$. As $y$ is the predecessor of $x$ in an $(M\oplus P)$-alternating path starting at $V^g_{M\oplus P}$, it implies that the edge $yx \notin M\oplus P$ and $yx \in M$. 
Now let us consider the path $P$. The path $P$ was the shortest $M$-alternating path starting at $V^g_{M}$.
Since $adist_{M}(y) < adist_{M}(x)$ and $xy\in P$ the path $P$ must visit the vertex $y$ before $x$. However, in such a case, by the definition of an alternating path starting at $V$, the edge going from $V$ to $U$  must be unmatched, a contradiction.

\smallskip
\noindent
Case 2. $x\in V$ and $y\in U$. As $y$ is a predecessor of $x$ in an $(M\oplus P)$-alternating path started at $V^g_{M \oplus P}$, it implies that $yx \notin M\oplus P$, consequently $yx \notin M$.
The path $P$ was the shortest $M$-alternating path started at $V^g_M$. Since $adist_{M}(y) < adist_{M}(x)$ and $xy \in P$ the path $P$ must first visit the vertex $y$ and then $x$. However, in such a case, from the definition of an alternating path starting at $V$, the edge going from $V$ to $U$ must be matched, a contradiction
\end{proof}

\section{The algorithm for finding a maximum $(f,g)$-semi-matching}

In this section we describe an algorithm for solving  the following problem:
\begin{problem}
 Given a bipartite graph $G=(U \cup V,E)$ and two mappings
$f: U \rightarrow \mathbb{N}$ and $g: V \rightarrow \mathbb{N}$.
 Find a maximum $(f,g)$-semi-matching of $G$.
\end{problem}

In order to simplify the notation, for an $(f,g)$-semi-matching $M$ of a bipartite graph $G = (U \cup V, E)$ and for each vertex of $u\in U \cup V$ we introduce the parameter $c_M(u)$ as follows: 
\begin{equation*}
c_M(u)= 
\begin{cases} 
f(u) - deg_M(u) & \text{if $u\in U$,} \\
g(u) - deg_M(u) & \text{if $u\in V$.}
\end{cases}
\end{equation*}

We denote by $M_{f,g}$-augmenting path an $M$-augmenting path with endvertices $u \in U$, $v\in V$, such that $c_M(u)>0$ and $c_M(v)>0$.

Our algorithm applies the same scheme as the well-known algorithm of Hopcroft-Karp \cite{HopcroftK73}. We start with an empty $(f,g)$-semi-matching $M$ and in each iteration we extend $M$ by several augmenting paths. The length of a shortest $M_{f,g}$-augmenting path increases after each iteration and each iteration of the algorithm consumes $O(m)$ time.

\begin{figure}[ht]
 \centerline{\includegraphics[scale=0.95]{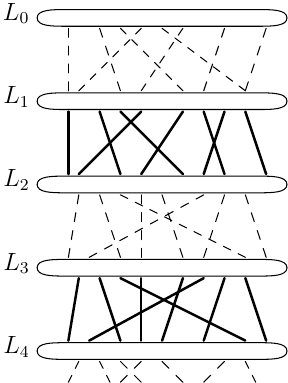}}
 \caption{The vertices of $G$ classified into layers}
 \label{fig2}
\end{figure}

One iteration of the algorithm finds a smallest number $t$ for which an $M_{f,g}$-augmenting path of length $t$ exists. Next, the algorithm extends $M$ by several augmenting paths in a single iteration, while there is an augmenting path of length $t$. More precisely:

\begin{enumerate}
 \item Let $L_0 = \{v \in V: c_M(v) > 0\}.$

\item 
In terms of  Breadth-First Search algorithm, classify vertices of $G$ into layers $L_1,L_2,\dots,L_n$ such that 
$L_i = \{v \in U\cup V: adist_M(v) = i \}$. This can be implemented as follows:\\
For each $i =0,2,4,\dots,2\lfloor n/2 \rfloor$ do\\
$L_{i+1} = \{u \in U: u\notin L_0,\dots,L_{i-1} \wedge \exists v\in L_i: uv \notin M  \}$\\
$L_{i+2} = \{v \in V: v\notin L_0,\dots,L_{i-1} \wedge \exists u\in L_{i+1} : uv \in M \}$\\

\item Let $t>0$ be a smallest odd number such that there exists $u\in L_t : c_M(u)>0$.
If no such $t$ exists, by Theorem~\ref{thmmaximum} there is no $M_{f,g}$-augmenting path. The algorithm stops and $M$ is a maximum $(f,g)$-semi-matching, otherwise continues by step 4.

\item For each vertex $u\in L_t$ while $c_M(u)>0$ do:\\
\begin{itemize} 
  \item[(i)] Find arbitrary $M_{f,g}$-augmenting path $P$ of length $t$ starting in $u$ such that $V(P) \subseteq L_0,L_1,\dots,L_t$.
  \item[(ii)] If such a path $P$ exists, set $M:= M\oplus P$ and recalculate values of $c_M$ along the path $P$.
\end{itemize}
\end{enumerate}

\begin{theorem} \label{increase}
 The length of the shortest augmenting path increases after each iteration of the algorithm.
\end{theorem}

\begin{proof}
An iteration which processes an $(f,g)$-semi-matching $M$ stops when there is no $M_{f,g}$-augmenting path consisting of vertices of $L_0 \cup L_1 \cup \dots \cup L_t$. It remains to prove, that after such an iteration there is no augmenting path of length $t$ in the graph $G$ (a path of length less than $t$ cannot appear due to Theorem~\ref{nondecrease} and the fact that all vertices in layers $L_1,L_2,\dots,L_{t-1}$ have zero capacity).

Suppose to the contrary, that after the iteration there is an $M'_{f,g}$-augmenting path $P = \{v_0, v_1, \dots, v_t\}$ of order $t$ in $G$. Since all the vertices of $V^g_{M'}$ are located in $L_0$, $v_0 \in L_0$. Since $P$ is an alternating path starting by a vertex of $L_0$, then $adist_{M'}(v_i) \le i$, for each $i=0,1,\dots,t$.
According to Theorem~\ref{nondecrease}, the value of $adist$ cannot decrease after iteration, i.e.\ $adist_{M}(v_i) \le adist_{M'}(v_i)$ for each $i=0,1,\dots,t$. Hence, each vertex of $P$ appears in $L_0 \cup L_1 \cup \dots L_t$ and such an augmenting path was not processed during the iteration of the algorithm, which is a contradiction.

\end{proof}

\subsection{The running time}

Let $n$ be the number of vertices in a given graph $G$ and $m$ be the number of its edges, assume that $m\ge n-1$ since isolated vertices can be erased from the graph in linear time. 

The algorithm starts with an empty $(f,g)$-semi-matching $M$ and then iterates several times until at least one augmenting path is found. In the search loop, the algorithm classifies the vertices into layers $L_0, L_1,\dots,L_t$ and modifies $M$ by augmenting paths using vertices of $L_0, L_1, \dots, L_t$. 
This step consumes $O(m)$ time, since each edge is manipulated at most once during one iteration.  
No more iteration is performed whenever no augmenting path was found in the actual loop. 

The key part of the complexity analysis is to enumerate the number of loops of the algorithm.
Let $s$ be the size of a maximum $(f,g)$-semi-matching $M^*$. After performing $\sqrt{s}$ iterations of the algorithm, according to Theorem~\ref{increase}, the shortest $M$-augmenting path consists of at least $\sqrt{s}$ vertices. According to Theorem~\ref{edgedis} there exist $s-|M|$ edge disjoint $M$-augmenting paths that can simultaneously extend $M$ to size $s$ and those paths consist only of edges of $M\cup M^*$. As each such a path must be of length at least $\sqrt{s}$ and $|M\cup M^*|$ is at most $2s$, these imply that $s-|M|\le 2\sqrt{s}$. 
Since in each loop the algorithm finds at least one augmenting path, the algorithm surely stops after at most $2\sqrt{s}$ loops.
Hence, the total number of performed loops is $O(\sqrt{s})$ and the algorithm runs in time $O(m \cdot \sqrt{s})$.

Moreover $s \le f(U)$ and $s \le g(V)$ and we get that the algorithm computes a maximum semi-matching in running time $O\left (m \cdot {\min\{\sqrt{f(U)}, \sqrt{g(V)}\}}\right)$.
For the case of $(f,1)$-semi-matching this gives the complexity upper bound $O(\sqrt{n} \cdot m)$.

To find an arbitrary $(f,g)$-quasi-matching  one can use the algorithm for maximum $(f,g)$-semi-matching problem which computes a maximum $(f,g)$-semi-matching $M$. Clearly, if $|M| < f(U)$ then no $(f,g)$-quasi-matching exists, otherwise $M$ is an $(f,g)$-quasi-matching. Moreover, for an $(f,g)$-quasi-matching we may assume $f(U) \ge g(V)$ (otherwise no $(f,g)$-quasi matching exists), we get the algorithm with running time 
$O(m \sqrt{g(V)})$.

\end{document}